\documentclass{article}
\usepackage[top=1.5in,left=1.0in,right=1.0in,bottom=1.0in,headheight=0in]{geometry}

\usepackage{amsmath} 
\usepackage{amssymb}  
\usepackage{bm}  
\usepackage{dsfont}
\usepackage{graphicx,color} 

\usepackage[inline]{enumitem}
\usepackage{epsfig} 
\usepackage{color}
\usepackage[tight,footnotesize]{subfigure}
\usepackage{cite}
\usepackage{url}
\usepackage{array}
\usepackage[svgnames,table]{xcolor}
\usepackage{dpfloat, booktabs}
\usepackage{pdflscape}
\usepackage{rotating}
\usepackage{longtable}
\usepackage{booktabs}
\usepackage{multirow}
\usepackage{algpseudocode}
\usepackage{algorithm}
\usepackage[colorlinks,linkcolor=Black,citecolor=Black,urlcolor=Black,breaklinks=true]{hyperref}
\newenvironment{rcases}
  {\left.\begin{aligned}}
  {\end{aligned}\right\rbrace}

\let\Pr\undefined 
\DeclareMathOperator{\Pr}{{\mathbf P}}        
\DeclareMathOperator*{\argmax}{arg\,max}
\DeclareMathOperator*{\argmin}{arg\,min}

\newtheorem{theorem}{Theorem}[section]

\newtheorem{lemma}[theorem]{Lemma}

\newtheorem{definition}{Definition}[section]

\newtheorem{example}{Example}

\newtheorem{remark}{Remark}[section]

\newenvironment{proof}[1][Proof]{\begin{trivlist}
\item[\hskip \labelsep {\bfseries #1}]}{\end{trivlist}}
\usepackage{enumitem}
\usepackage{epstopdf}

\newcommand{\cupdot}{\mathbin{\mathaccent\cdot\cup}}

\title{\LARGE \bf
Distributed Learning in the Presence of Disturbances}
\author{Chithrupa Ramesh, Marius Schmitt and John Lygeros
\thanks{This work was supported by the EU project SPEEDD (FP7-ICT 619435).}%
\thanks{C. Ramesh, M. Schmitt and J. Lygeros are with the Automatic Control Lab,
        Electrical Engineering, ETH, Zurich,
        Switzerland. {\tt\small \{rameshc,schmittm,lygeros\}@control.ee.ethz.ch}}%
}

\begin{document}

\maketitle \thispagestyle{empty} \pagestyle{empty}

\begin{abstract}
We consider a problem where multiple agents must learn an action profile that maximises the sum of their utilities in a distributed manner. The agents are assumed to have no knowledge of either the utility functions or the actions and payoffs of other agents. These assumptions arise when modelling the interactions in a complex system and communicating between various components of the system are both difficult. In~\cite{Marden2014}, a distributed algorithm was proposed, which learnt Pareto-efficient solutions in this problem setting. However, the approach assumes that all agents can choose their actions, which precludes disturbances. In this paper, we show that a modified version of this distributed learning algorithm can learn Pareto-efficient solutions, even in the presence of disturbances from a finite set. We apply our approach to the problem of ramp coordination in traffic control for different demand profiles.
\end{abstract}

\section{Introduction} \label{Sec:Intro}

In complex systems, modelling the interactions between various components and their relationship to the system performance is not an easy task. This poses a challenge while designing controllers for such systems, as most design methods require a model of the system. Even when considerable effort has been expended in identifying suitable models for such systems, utilising these models to design online controllers is not always easy. This is because collecting measurements of a complex system, computing control signals using complex algorithms and applying these controls to actuators across the system is communication intensive and computationally demanding. The resulting delays are not well suited to the control of real-time complex systems.

An example is the real-time control of freeway traffic, where often traffic models are highly nonlinear and methods to design controllers using these models do not scale well~\cite{Allsop2008}. Furthermore, to use these models, the traffic flow from every segment of the freeway must be measured and collected, and the control signals must be delivered to the ramps on the freeway. To reduce the communication and computation burden, distributed controllers that act on mostly local information are required.

One approach is to use a distributed randomised algorithm to explore the policy space and learn the optimal actions. Recently, a distributed learning algorithm has been proposed in~\cite{Marden2014} where agents learn action profiles that maximise the system welfare. This algorithm is payoff-based, and the agents require no prior knowledge of either the utility functions or the actions and payoffs of other agents. An implicit assumption in this approach is that every agent that influences the utility can choose its actions. In reality, there might always be disturbances which cannot be chosen in a desired manner. In this paper, we extend this approach to include the effects of disturbances.

Our main contribution is a modification to the algorithm in~\cite{Marden2014} to deal with disturbances. We show that agents learn Pareto-efficient solutions in a distributed manner using our algorithm, even in the presence of disturbances from a finite set. We verify the theoretical results on a small example. In this case, all assumptions can be verified and strong convergence guarantees can be given. To demonstrate versatility of the approach, we also apply the results to a realistic coordination problem motivated by freeway traffic control. We use our newly developed algorithm to learn a high-level coordination strategy for a ramp metering problem with promising results, using simulation parameters and traffic demand data from a real-world use case.

The learning rule used in this paper is related to the trial and error learning procedure from~\cite{Young2009} and its cognates~\cite{Pradelski2012,Marden2014}. These papers proposed algorithms that learnt Nash equilibria~\cite{Young2009}, Pareto efficient equilibria~\cite{Pradelski2012} and Pareto-efficient action profiles~\cite{Marden2014}, respectively. Convergence guarantees for the latter were presented in~\cite{Menon2013a}. Restrictions on the payoff structure, which are required for the result in~\cite{Marden2014} to hold, were eliminated through the use of explicit communication in~\cite{Menon2013b}. We also draw on the analysis of deliberate experimentation using the theory of regular perturbed Markov processes from~\cite{Young1993}.

This paper is organised as follows: We describe the algorithm in Section~\ref{Sec:ProbForm}, and present known results in Section~\ref{Sec:Preliminaries}. Our main result is presented in Section~\ref{Sec:MainResults} and illustrated on a few examples in Section~\ref{Sec:Examples}. The conclusion is in Section~\ref{Sec:Conclusions}.

\section{Problem Formulation} \label{Sec:ProbForm}

We consider a set of agents $N := \{1,\dots,n\}$, each with a finite action set $\mathcal{A}_i$ for $i \in N$. The disturbance is modelled as an independent and identically distributed (i.i.d.) process $w_{k}$, which takes values from a finite set $\mathbb{W}$ according to a probability distribution $\Pr_{w}$ that is fully supported on $\mathbb{W}$. Given an action profile $a \in \mathcal{A}$, where $\mathcal{A} := \mathcal{A}_1 \times \dots \times \mathcal{A}_n$, and a disturbance $w \in \mathbb{W}$, the payoff for each agent is $u_i(a,w)$. The payoffs are generated by utility functions $\mathcal{U}_i: \mathcal{A} \times \mathbb{W} \rightarrow [0,1)$ whose functional forms are unknown to the agents. The welfare of the network of agents is $\mathcal{W}(a,w) = \sum_{i \in N} \mathcal{U}_i(a,w)$.

The agents play a repeated game; in the $k^{\textrm{th}}$ iteration, each agent chooses its action $a_{i,k}$ with probability $p_{i,k} \in \Delta(\mathcal{A}_i)$, where $\Delta(\mathcal{A}_i)$ is the simplex of distributions over $\mathcal{A}_i$. The strategy $p_{i,k}$ is completely uncoupled or pay-off based, i.e., $p_{i,k} = \psi_i(\{a_{i,\tau},u_{i,\tau}(a_{\tau},w_{\tau})\}_{\tau=0}^{k-1})$. In other words, an agent does not know the actions or payoffs of any other agent in the network.

Each agent maintains an internal state $z_{i,k} := [\bar{a}_{i,k}, \bar{u}_{i,k}, m_{i,k}]$ in the $k^{\rm th}$ iteration, where $\bar{a}_{i,k} \in \mathcal{A}_i$ is the baseline action, $\bar{u}_{i,k}$ is the corresponding baseline utility that lies in the range of $\mathcal{U}_i$ and $m_{i,k} \in \{\mathcal{C},\mathcal{D}\}$ is the mood variable that connotes whether the agent is content or discontent. The state $z_k := \{z_{1,k},\dots,z_{n,k}\}$ lies in the finite state space $\mathbb{Z}$.

The algorithm is initialised with all agents setting their moods to discontent, i.e., $m_{i,0} = \mathcal{D}$ for $i \in N$. An experimentation rate $0 < \varepsilon < 1$ is fixed and a constant $c > n$ is selected. Then,
each agent selects an action $a_{i,k}$ according to its mood and the corresponding probabilistic rule:
\begin{equation} \label{Eq:AgentDynamics}
\begin{aligned}
&m_{i,k} = \mathcal{C}: \; p_i(a_{i,k}) &= \begin{cases}
\frac{\varepsilon^c}{|\mathcal{A}_i|-1} & a_{i,k} \neq \bar{a}_{i,k} \\
1-\varepsilon^c & a_{i,k} = \bar{a}_{i,k}
\end{cases} \\
&m_{i,k} = \mathcal{D}: \; p_i(a_{i,k}) &= \frac{1}{|\mathcal{A}_i|} \quad \forall \quad a_{i,k} \in \mathcal{A}_i
\end{aligned}
\end{equation}
The agents choose their strategies based on their moods. A content agent selects its baseline action with high probability and experiments by choosing other actions with low probability. A discontent agent selects an action with uniform probability.

Each agent plays the action it has selected and receives a payoff $u_{i,k}(a_k,w_k)$, which it uses to update its state as
\begin{align}
&z_{i,k+1} = \begin{cases}
z_{i,k} & \begin{aligned}
m_{i,k} &= \mathcal{C}, a_{i,k} = \bar{a}_{i,k}, \\
|u_{i,k} &- \bar{u}_{i,k}| \le \rho
\end{aligned} \\
\begin{rcases}
z_{\mathcal{C}} & \; \text{w.p.} \; p_{\mathcal{C}} \\
z_{\mathcal{D}} & \; \text{w.p.} \; 1-p_{\mathcal{C}}
\end{rcases}
& \rm{otherwise}
\end{cases} \label{Eq:StateDynamics}
\end{align}
where $z_{\mathcal{C}} = \left[a_{i,k}, u_{i,k}, \mathcal{C} \right]$, $z_{\mathcal{D}} = \left[a_{i,k}, u_{i,k}, \mathcal{D} \right]$, $p_{\mathcal{C}} = \varepsilon^{1-u_{i,k}}$ and $\rho$ is the maximum deviation in the payoffs due to the disturbance process $w$, as defined in~\eqref{Eq:RhoDef}. The state update also depends on the mood of the agent. A content agent that chose to play its baseline action and received a payoff within the interval $u_{i,k} \in [\bar{u}_{i,k}-\rho,\bar{u}_{i,k}+\rho]$ retains its state. In all other cases, the state is updated to the played action and received payoff, and the mood is set to content or discontent with high probabilities for high or low payoffs, respectively. Thus, a content agent must receive a payoff outside an interval $\pm \rho$ of its baseline payoff to reevaluate its mood or change its state. This interval rule renders an agent insensitive to small changes in the payoff.

The variable $\rho$ is defined as the maximum deviation in the payoffs received by any agent $i \in N$ for every action profile $a \in \mathcal{A}$ and every pair of disturbances $w_1, w_2 \in \mathbb{W}$, i.e.,
\begin{equation} \label{Eq:RhoDef}
\begin{aligned}
\rho := \argmin_{r \in \mathbb{R}} &\{ |u_i(a,w_1) - u_i(a,w_2)| \le r \}
&\forall \; i \in N , \; \forall \; a \in \mathcal{A} \; \textrm{ and } \; \forall \; w_1, w_2 \in \mathbb{W} \; .
\end{aligned}
\end{equation}

We are interested in identifying the set of states the above algorithm converges to. A necessary condition for this algorithm to function as desired is the interdependence property, stated below.
\begin{definition} \label{Def:Interdependence}
An $n$-person game is \emph{interdependent} if for every action profile $a \in \mathcal{A}$, every disturbance $w \in \mathbb{W}$ and every proper subset of agents $J \subset N$, there exists an agent $i \subset N \setminus J$, a choice of actions $a'_J \in \prod_{j \in J} \mathcal{A}_j$ and a disturbance $w' \in \mathbb{W}$, such that
\begin{equation} \label{Eq:Interdependence}
\left| \mathcal{U}_i(a'_J,a_{-J},w') - \mathcal{U}_i(a_J,a_{-J},w) \right| > \rho
\end{equation}
\end{definition}
This property ensures that the set of agents cannot be divided into two mutually non-interacting groups, and that a discontent agent always has recourse to actions that influence the utilities of other content agents despite the algorithm's insensitivity to the interval $[\bar{u}_i-\rho,\bar{u}_i+\rho]$ in~\eqref{Eq:StateDynamics}.

\begin{remark}[A remark on the state space $\mathbb{Z}$: ]
The state $z_{k}$ is an aggregation of the states $z_{i,k} := [\bar{a}_{i,k}, \bar{u}_{i,k}, m_{i,k}]$ of each of the agents. Thus, one would expect the cardinality of the state space to be $|\mathbb{Z}| = 2^N |\mathcal{A}| |\mathbb{W}|$, because the payoffs obtained are completely determined by the choice of the actions and the disturbance. However, the interval rule in~\eqref{Eq:StateDynamics} results in more states becoming reachable and $|\mathbb{Z}| \le 2^{N} |\mathcal{A}|^{2} |\mathbb{W}|$. The exact number of states depends on the payoffs. For the proofs presented in this paper, we define the state space in terms of the states reachable from the initial point of our algorithm, i.e., $\mathbb{Z} := \{ z : \exists \tau > 0 \; \text{s.t.} \; \Pr(z_{\tau}=z | z_{0}) > 0 \}$, where $z_{0}$ is any state where all agents are discontent.
\end{remark}

%
\section{Preliminaries} \label{Sec:Preliminaries}

We briefly outline Young's result on regular perturbed Markov processes~\cite{Young1993}. 
Consider the Markov processes on a state space $\mathbb{X}$ with transition matrices $\Pr^{0}$ and $\Pr^\epsilon$, where a finite-valued $\epsilon > 0$ measures the noise level. The Markov chain induced by $\Pr^0$ describes some basic evolutionary process such as best response dynamics, while the chain induced by $\Pr^{\epsilon}$ represents the perturbed process obtained by introducing mistakes or experiments. This notion is formalised as follows.

\begin{definition} \label{Def:RegPerturb}
A family of Markov processes $\Pr^{\epsilon}$ is called a regular perturbation of a Markov chain with transition matrix $\Pr^0$ if it satisfies the following conditions:
\begin{enumerate}[label=\roman*.]
    \item $\Pr^{\epsilon}$ is aperiodic and irreducible for all finite $\epsilon > 0$.
    \item $\lim_{\epsilon \rightarrow 0} \Pr^{\epsilon}_{xy} = \Pr^0_{xy}$, $\forall x,y \in \mathbb{X}$.
    \item If $\Pr^{\epsilon}_{xy} > 0$ for some $\epsilon$, then $\exists$ $r(x,y) \ge 0$, called the \emph{resistance} of the transition $x \rightarrow y$, such that
                    \begin{equation} \label{Eq:Resistance}
                    0 < \lim_{\epsilon \rightarrow 0} \epsilon^{-r(x,y)} \Pr^{\epsilon}_{xy} < \infty
                    \end{equation}
\end{enumerate}
\end{definition}
Property~i ensures that there is a unique stationary distribution for all finite $\epsilon > 0$. Property~ii ensures that the perturbed process converges to the unperturbed process in the limit as $\epsilon \rightarrow 0$. Property~iii states that a transition $x \rightarrow y$ is either impossible under $\Pr^{\epsilon}$ or it occurs with a probability $\Pr^{\epsilon}_{xy}$ of order $\epsilon^{r(x,y)}$ for some unique, real $r(x,y)$ in the limit as $\epsilon \rightarrow 0$. Note that $r(x,y)=0$ if and only if $\Pr^0_{xy} > 0$. Thus, the transitions of resistance zero are the same as the transitions that are feasible under $\Pr^0$. 

\begin{definition} \label{Def:StochStableState}
A state $x \in \mathbb{X}$ is said to be stochastically stable if $\mu^0_x > 0$, where $\mu^0$ is a stationary distribution of $\Pr^0$.
\end{definition}

We are interested in characterizing the limiting distribution $\mu^0$ of $\Pr^{0}$ through its support, or the set of stochastically stable states. To do this, we define two directed graphs. The first graph $G := (\mathbb{X},\mathbb{E}_G)$ has as vertex set the set of states $\mathbb{X}$, and as directed edge set $\mathbb{E}_G := \{ x \rightarrow y \; | \; \Pr^{\epsilon}_{xy} > 0 , \; x,y \in \mathbb{X}\}$. Thus, a directed edge exists in $G$ only if a single transition under $\Pr^{\epsilon}$ gets us from state $x$ to $y$, for all values of $\epsilon \ge 0$. Finally, $r(x,y)$ in~\eqref{Eq:Resistance} defines the weight or resistance of this directed edge in $G$.

To define the second graph, we first enumerate the recurrence classes of $\Pr^0$ as $X_1,\dots,X_L$. Then, we can define the resistance between two classes as the minimum resistance between any two states belonging to these classes, i.e.,
\begin{equation} \label{Eq:ResistanceRecClass}
r_{\ell_1 \ell_2} := \min_{x \in X_{\ell_1}, y \in X_{\ell_2}} r(x,y) , \quad \quad \text{for} \quad \ell_1, \ell_2 \in \{1,\dots,L\} \; .
\end{equation}
Note that there is at least one path from every class to every other because $\Pr^{\epsilon}$ is irreducible. We now define the second graph as $\mathcal{G} := (\{1,\dots,L\}, \mathbb{E}_{\mathcal{G}})$. This graph has as vertex set the set of indices of the recurrence classes of $\Pr^0$, and as edge set the set of directed edges between members of the recurrence classes. Also, $r_{\ell_1 \ell_2}$ defines the resistance or weight of this directed edge.

\begin{definition} \label{Def:StochPotential}
Let an $\ell$-tree in $\mathcal{G}$ be a spanning sub-tree of $\mathcal{G}$, such that for every vertex $\ell' \neq \ell$, there exists exactly one directed path from $\ell'$ to $\ell$. Then, the stochastic potential $\gamma_\ell$ of the recurrence class $X_\ell$ is defined as
\begin{equation} \label{Eq:StochPotential}
\gamma_\ell := \min_{T \in \mathcal{T}_\ell} \sum_{(a,b) \in T} r_{a b}
\end{equation}
where $\mathcal{T}_\ell$ is the set of all $\ell$-trees in $\mathcal{G}$.
\end{definition}

We can now state Young's result for perturbed Markov processes~\cite{Young1993}.
\begin{theorem}[Theorem~$4$ from~\cite{Young1993}] \label{Thm:Young93}
Let $\Pr^0$ be a time-homogenous Markov process on the finite state space $\mathbb{X}$ with recurrence classes $X_1, \dots, X_L$. Let $\Pr^{\epsilon}$ be a regular perturbation of $\Pr^0$, and let $\mu^{\epsilon}$ be its unique stationary distribution for every small positive $\epsilon$. Then,
\begin{enumerate}[label=\roman*.]
\item as $\epsilon \rightarrow 0$, $\mu^{\epsilon}$ converges to a stationary distribution $\mu^0$ of $\Pr^{0}$,
\item the recurrence class $X_{\ell^{*}}$, with stochastic potential $\gamma_{\ell^{*}} := \min_{\ell \in \{1,\dots,L\}} \gamma_{\ell}$, contains the stochastically stable states $\{x \in \mathbb{X}: \mu^0_x > 0 \}$.
\end{enumerate}
\end{theorem}

\section{Learning Pareto-efficient solutions} \label{Sec:MainResults}

We begin by establishing that Young's result applies to our system, resulting in a distributed algorithm for Pareto-efficient learning in the presence of disturbances. To prove this result, we enumerate the recurrence classes of $\Pr^0$ and the resistances between the classes. We use these values to identify the structure of the tree with minimum stochastic potential.

\subsection{Main Result}

For $\varepsilon = 0$, the transition matrix $\Pr^{0}$ corresponds to an unperturbed Markov process, and we begin by showing that $\Pr^{\varepsilon}$ is a regular perturbation on $\Pr^0$.
\begin{lemma} \label{Lem:RegPert}
The Markov process with transition matrix $\Pr^{\varepsilon}$ is a regular perturbation on $\Pr^0$.
\end{lemma}

The proof is presented in Appendix~\ref{App:ProofLemmas}. Next, we use Young's result from Theorem~\ref{Thm:Young93} to obtain the distributed learning outcome stated below.
\begin{theorem} \label{Thm:ParetoEffLearning}
Let $G$ be an interdependent $n$-person game on a finite joint action space $\mathcal{A}$, subject to i.i.d. disturbances from a finite set $\mathbb{W}$. Under the dynamics defined in~\eqref{Eq:AgentDynamics}--\eqref{Eq:StateDynamics}, a state $z = [\bar{a},\bar{u},m] \in \mathbb{Z}$ is stochastically
stable if and only if the following conditions are satisfied:
\begin{enumerate}[label=\roman*.]
\item The action profile $\bar{a}$ maximises the network welfare, i.e.,
\begin{equation} \label{Eq:WelfareMax}
(\bar{a},\bar{w}) \in \argmax_{a \in \mathcal{A}, w \in \mathbb{W}} \mathcal{W} = \argmax_{a \in \mathcal{A}, w \in \mathbb{W}} \sum_{i \in N} \mathcal{U}_i(a,w)
\end{equation}
\item The benchmark actions and payoffs are aligned for the maximising disturbance, i.e., $\bar{u}_i = \mathcal{U}_i(\bar{a},\bar{w})$. 
\item The mood of each agent is content. 
\end{enumerate}
\end{theorem}
We present the proof for this theorem in the next section.

\subsection{Recurrence Classes}

The states $z \in \mathbb{Z}$ can be classified into three categories: states where all agents are content or discontent and states where some agents are content and others discontent. By inspecting the algorithm in~\eqref{Eq:AgentDynamics}--\eqref{Eq:StateDynamics}, it is easy to see that as $\varepsilon \rightarrow 0$ the former states can be recurrent, but not the latter. We formalise this notion below, by defining the recurrence classes $\mathds{D}$ and $\mathds{C}^m$ for $0 \le m < n$ and showing that there are no other recurrence classes.

\noindent \textbf{Discontent Class $\mathds{D}$:} The states in this recurrence class correspond to those where all agents are discontent.
    \begin{equation} \label{Eq:DiscontentClass}
        \begin{aligned}
            \mathds{D} := \Big \{ z \in \mathbb{Z} \quad \big | \quad m_i(z) &= \mathcal{D} , \; \forall i \in N \Big \} 
        \end{aligned}
    \end{equation}
Note that the payoffs and action profiles are aligned, i.e., $\bar{u}_i(z) = \mathcal{U}_i(\bar{a}(z),w)$, $\forall z \in \mathds{D}$, $\forall i \in N$ and for some $w \in \mathbb{W}$. Also, corresponding to each action profile and disturbance pair $(a,w) \in \mathcal{A} \times \mathbb{W}$, there is a discontent state in this recurrence class.

States containing only content agents can be categorised further into $n$ classes, $\mathds{C}^m$ for $0 \le m < n$, as follows. \newline
\noindent \textbf{$0^{\textrm{th}}$-Content Class $\mathds{C}^0$:} This recurrence class contains singleton states where all agents are content, and where the payoffs of all agents are aligned with the action profile for some value of the disturbance $w \in \mathbb{W}$, while satisfying the interval rule in~\eqref{Eq:StateDynamics} for all other values of the disturbance. Let $B_{i}$ denote the set of states that satisfy these conditions on the payoffs of the $i^{\textrm{th}}$ agent:
   \begin{equation} \label{Eq:PayoffCond}
    \begin{aligned}
    B_{i} := \Big \{ z \in \mathbb{Z} \big | &\bar{u}_i(z) = \mathcal{U}_i(\bar{a}(z),w), \text{for some} \ w \in \mathbb{W} , \\
            &|\bar{u}_i(z) - \mathcal{U}_i(\bar{a}(z),\tilde{w})| \le \rho , \; \; \forall \ \tilde{w} \in \mathbb{W} \Big \} .
    \end{aligned}
   \end{equation}
    Then, the recurrence class $\mathds{C}^0$ is defined as
    \begin{equation} \label{Eq:ContentClass0}
        \begin{aligned}
            \mathds{C}^0 := \Big \{ z \in \mathbb{Z} \quad \big | \quad m_i(z) = \mathcal{C} , 
            z \in B_{i} , \forall i \in N \Big \} \; .
        \end{aligned}
    \end{equation}
From the definition of $\rho$ in~\eqref{Eq:RhoDef}, we know that corresponding to each action profile and disturbance pair $(a,w) \in \mathcal{A} \times \mathbb{W}$, there is a state in this recurrence class with payoffs satisfying~\eqref{Eq:PayoffCond}. 

There might also be states where the payoffs of all agents are not aligned with the action profile for any single value of the disturbance $w \in \mathbb{W}$. Some of these states can be recurrent, and belong to the classes defined below. \newline
\noindent \textbf{$1^{\textrm{st}}$-Content Class $\mathds{C}^1$:} Suppose that a proper subset of agents $J_{1} \subset N$ from a state $z' \in \mathds{C}^{0}$ experiment with different actions despite being content, and become content with their new utilities. If the rest of the agents $j_{0} \in J_{0} = N \setminus J_{1}$ do not notice this change, because their new utilities lie within the interval $[\bar{u}_{j_{0}}(z')-\rho,\bar{u}_{j_{0}}(z')+\rho]$ for all values of the disturbance, then the agents find themselves in a state $z$ in a recurrence class $\mathds{C}^{1}$. 
    \begin{equation} \label{Eq:ContentClass1}
        \begin{aligned}
            \mathds{C}^1 := \Big \{ z \in \mathbb{Z} \; \big | \; &m_i(z) = \mathcal{C} , \forall i \in N , \\
            &\exists (J_{0},J_{1}) \; \text{s.t.} \; J_{0} \cupdot J_{1} = N , 
            z \in B_{j_{1}} , \quad \forall \; j_{1} \in J_{1} , \\
            & \exists z' \in \mathds{C}^{0} \; \text{s.t.} \; z_{j_{0}} = z'_{j_{0}} , 
            |\bar{u}_{j_{0}}(z') - \mathcal{U}_{j_{0}}(\bar{a}(z),\tilde{w})| \le \rho , 
            \quad \forall \; \tilde{w} \in \mathbb{W} , \forall \; j_{0} \in J_{0} \; \Big \} ,
        \end{aligned}
    \end{equation}
where the symbol $\cupdot$ denotes a disjoint union of the subsets.

A subset of agents from a state in $\mathds{C}^{1}$ could experiment and find themselves in a state in a recurrence class $\mathds{C}^{2}$. In general, states in the recurrence class $\mathds{C}^m$ can be reached from a state in $\mathds{C}^{m-1}$, following a similar procedure. 
The recurrence class $\mathds{C}^m$ is defined below. \newline
\noindent \textbf{$m^{\textrm{th}}$-Content Class $\mathds{C}^m$:} These recurrence classes contain singleton states where all agents are content, and where the agents can be divided into $m+1$ mutually disjoint subsets $J_{0},\dots,J_{m}$, such that the utilities of the agents within each subset are aligned with an action profile for some value of the disturbance. 
    \begin{equation} \label{Eq:ContentClassm}
        \begin{aligned}
            \mathds{C}^m := \Big \{ z \in \mathbb{Z} \; \big | \; &m_i(z) = \mathcal{C} , \forall i \in N , \\
            &\exists (J_{0},\dots, J_{m}) \; \text{s.t.} \; \cupdot_{l=0}^{m} J_{l} = N , 
            z \in B_{j_{m}} , \quad \forall \; j_{m} \in J_{m} , \\
            & \exists z' \in \mathds{C}^{m-1} \; \text{s.t.} \; z_{j_{\ell}} = z'_{j_{\ell}} , 
            |\bar{u}_{j_{\ell}}(z') - \mathcal{U}_{j_{\ell}}(\bar{a}(z),\tilde{w})| \le \rho , 
            \quad \forall \; \tilde{w} \in \mathbb{W} , \forall \; j_{\ell} \in N \setminus J_{m} \; \Big \} .
        \end{aligned}
    \end{equation}
There can be at most $n$ disjoint subsets from a set of $n$ agents, and hence $m < n$. 
Clearly, there might be many states in $\mathbb{Z}$, where the baseline payoffs and actions satisfy some, but not all, of the above conditions for classes $\mathds{C}^m$, $1 \le m < n$. These states are not recurrent, as we show below.

\begin{lemma} \label{Lem:RecurrenceClasses}
The recurrence classes corresponding to the $n$-person interdependent game described by~\eqref{Eq:AgentDynamics}--\eqref{Eq:StateDynamics} are $\mathds{D}$, and the singletons in $\mathds{C}^{0}$ and $\mathds{C}^{m}$, for $0 < m < n$, as defined in \eqref{Eq:DiscontentClass}--\eqref{Eq:ContentClassm}, respectively.
\end{lemma}
The proof of this result is presented in Appendix~\ref{App:ProofLemmas}.

\subsection{Resistances and Trees}

\begin{table}[tb]
\renewcommand{\arraystretch}{1.3}
\caption{Resistances Between Recurrence Classes}
\label{Tab:Resistances}
\centering
\begin{tabular}{| >{$} r <{$} | >{\centering\arraybackslash}m{3cm} || >{\centering\arraybackslash}m{4cm} |}
\hline
\text{No.} & Path & Resistance Relationship \\
\hline \hline
1 & $\mathds{D} \rightarrow \mathds{C}^0$ & $r_{d z^0} = \sum_{i \in N} 1 - \bar{u}_i(z^0)$ \\ 
2 & $\mathds{D} \rightarrow \mathds{C}^m$ & $r_{d z^m} = \min_{z^0 \in \mathds{C}^0} r_{d z^0} + r_{z^0 z^m}$ \\ 
3 & $\mathds{C}^0 \rightarrow \mathds{D}$ & $r_{z^0 d} = c$ \\ 
4 & $\mathds{C}^m \rightarrow \mathds{D}$ & $r_{z^m d} = c$ \\ 
5 & $\mathds{C}^0 \rightarrow \mathds{C}^0$ & $c \le r_{z^0_1 z^0_2} \le 2c$ \\ 
6 & $\mathds{C}^m \rightarrow \mathds{C}^m$ & $c \le r_{z^m_1 z^m_2} \le \bar{r}_m$ \\
7 & $\mathds{C}^l \rightarrow \mathds{C}^m$, \small{$0 \le l,m <n$, $m \neq l$} & $|m-l| c \le r_{z^m_1 z^m_2} \le \bar{r}_m$ \\
\hline
\end{tabular}
\end{table}

Transitions can occur between all three recurrence class types, namely $\mathds{D} \rightarrow \mathds{C}^0$ and vice versa, $\mathds{D} \rightarrow \mathds{C}^{m}$ and vice versa, and $\mathds{C}^l \rightarrow \mathds{C}^m$ for $0 \le l,m < n$ and $l \neq m$. In addition, the singleton states in $\mathds{C}^0$ and $\mathds{C}^m$ can transition to other singleton states within the same classes. All these transitions are enumerated along with the corresponding resistances in Table~\ref{Tab:Resistances}. In this table, we use $d \in \mathds{D}$, $z^0_{\cdot} \in \mathds{C}^0$ and $z^m_{\cdot} \in \mathds{C}^m$ to denote states in the respective recurrence classes. Some of the entries contain the term $\bar{r}_{m}$, which is given by
\begin{equation}
\bar{r}_{m} = mc^{2} + \frac{4+m-m^{2}}{2} c - \frac{m(m+1)}{2} , \; 0 < m < n \; . \label{Eq:rmm}
\end{equation}
The calculations for the entries in Table~\ref{Tab:Resistances} are presented in Appendix~\ref{App:ResistanceCalc}. We can now compute the stochastic potential of a state in $\mathds{C}^{0}$ and show that a minimum potential tree is rooted at a singleton in $\mathds{C}^{0}$.
\begin{lemma} \label{Lem:StochPotC0}
The stochastic potential of a state $z^{0} \in \mathds{C}^{0}$ is
\begin{equation} \label{Eq:StochPotC0}
\gamma(z^{0}) = c \left( \sum_{m=0}^{n-1} |\mathds{C}^{m}| -1 \right) + \sum_{i \in N} \left( 1 - \bar{u}_{i}(z^{0}) \right) \; .
\end{equation}
\end{lemma}

\begin{lemma} \label{Lem:NotStochStable}
The states in the recurrence class $\mathds{D}$ and the singletons $\mathds{C}^{m}$, for $0 < m < n$, are not stochastically stable.
\end{lemma}

The proofs for both Lemmas are presented in Appendix~\ref{App:ProofLemmas}.

\subsection{Proof of the Main Result}
We now present the proof of Theorem~\ref{Thm:ParetoEffLearning}. \newline
\begin{proof}
The stochastically stable states are contained in the recurrence class of $\Pr^{0}$ with minimum stochastic potential (from Theorem~\ref{Thm:Young93}). From Lemma~\ref{Lem:NotStochStable}, we also know that the recurrence class with minimum stochastic potential is rooted at a singleton in $\mathds{C}^{0}$.

Lemma~\ref{Lem:StochPotC0} gives us the minimum stochastic potential as
\begin{equation*}
\gamma(z^{0,*}) = \min_{z_{0} \in \mathds{C}^{0}} c ( \sum_{m=0}^{n-1} |\mathds{C}^{m}| -1 ) + \sum_{i \in N} ( 1 - \bar{u}_{i}(z^{0}) )
\end{equation*}
Thus, the action profile corresponding to the state $z^{0,*}$ must satisfy $\bar{a}(z^{0,*}) \in \argmax_{a \in \mathcal{A}, w \in \mathbb{W}} \sum_{i \in N} \mathcal{U}_{i}(a,w)$.

From the definition of the recurrence class $\mathds{C}^{0}$ in~\eqref{Eq:ContentClass0}, we obtain statements~ii and~iii of the theorem. 
\end{proof}

\begin{table}[tb]
\renewcommand{\arraystretch}{1.2}
\caption{Payoffs in Example~\ref{Ex:FixedEps}} \label{Tab:SimpleExample}
\centering
\begin{tabular}{| >{$ \{} r <{\} $} | >{$}l<{$} | >{$}l<{$} || >{$ \{} r <{\} $} | >{$}l<{$} | >{$}l<{$} |}
\hline
a_1,a_2,w & u_1 & u_2 & a_1,a_2,w & u_1 & u_2 \\
\hline \hline
0,0,0 & 0.30 & 0.40 & 1,1,0 & 0.80 & 0.90 \\
0,0,1 & 0.40 & 0.30 & 1,1,1 & 0.90 & 0.80 \\
0,1,0 & 0.20 & 0.10 & 2,0,0 & 0.65 & 0.55 \\
0,1,1 & 0.10 & 0.20 & 2,0,1 & 0.55 & 0.65 \\
1,0,0 & 0.60 & 0.50 & 2,1,0 & 0.75 & 0.85 \\
1,0,1 & 0.50 & 0.60 & 2,1,1 & 0.85 & 0.75 \\
\hline
\end{tabular}
\end{table}
\begin{table*}[!tb]
\renewcommand{\arraystretch}{1.2}
\caption{Fraction of occurrence of states in Example~\ref{Ex:FixedEps}} \label{Tab:Hist3}
\centering
\begin{tabular}{| >{\centering\arraybackslash}m{4cm} | m{3cm} || >{\centering\arraybackslash}m{4cm} | m{3cm} |}
\hline
State $z$ & Normalised Number of Instances & State $z$ & Normalised Number of Instances \\
\hline \hline
$\{ [0,0.10,\mathcal{C}] , [1,0.20,\mathcal{C}] \}$ & $0.0009$ & $\{ [1,0.60,\mathcal{C}] , [0,0.50,\mathcal{C}] \}$ & $0.0013$ \\
$\{ [1,0.50,\mathcal{C}] , [0,0.60,\mathcal{C}] \}$ & $0.0016$ & $\{ [1,0.50,\mathcal{C}] , [0,0.50,\mathcal{C}] \}$ & $0.0004$ \\
$\mathbf{\{ [1,0.80,\mathcal{C}] , [1,0.90,\mathcal{C}] \}}$ & $\mathbf{0.9532}$ & $\{ [1,0.90,\mathcal{C}] , [1,0.90,\mathcal{C}] \}$ & $0.0123$ \\
$\{ [1,0.90,\mathcal{C}] , [1,0.80,\mathcal{C}] \}$ & $0.0032$ & $\{ [1,0.90,\mathcal{C}] , [1,0.85,\mathcal{C}] \}$ & $0.0034$ \\
$\{ [1,0.8,\mathcal{C}] , [1,0.85,\mathcal{C}] \}$ & $0.0002$ & $\{ [2,0.65,\mathcal{C}] , [0,0.65,\mathcal{C}] \}$ & $0.0014$ \\
$\{ [2,0.55,\mathcal{C}] , [0,0.65,\mathcal{C}] \}$ & $0.0007$ & $\{ [2,0.65,\mathcal{C}] , [0,0.55,\mathcal{C}] \}$ & $0.0004$ \\
$\{ [2,0.55,\mathcal{C}] , [0,0.55,\mathcal{C}] \}$ & $0.0002$ & $\{ [2,0.65,\mathcal{C}] , [0,0.60,\mathcal{C}] \}$ & $0.0007$ \\
$\{ [2,0.85,\mathcal{C}] , [1,0.85,\mathcal{C}] \}$ & $0.0022$ & $\{ [2,0.75,\mathcal{C}] , [1,0.85,\mathcal{C}] \}$ & $0.0015$ \\
$\{ [2,0.85,\mathcal{C}] , [1,0.75,\mathcal{C}] \}$ & $0.0059$ & $\{ [2,0.75,\mathcal{C}] , [1,0.80,\mathcal{C}] \}$ & $0.0088$ \\
$\{ [2,0.85,\mathcal{C}] , [1,0.80,\mathcal{C}] \}$ & $0.0006$ &  & \\
\hline
\end{tabular}
\end{table*}

\section{Examples} \label{Sec:Examples}

We present a simple example of a two-agent interdependent game to illustrate the results of Theorem~\ref{Thm:ParetoEffLearning}, and then apply this method to the ramp coordination problem.

\begin{example} \label{Ex:FixedEps}
Consider a simple game $G_2$ with $n=2$ agents. The action sets, disturbance set and payoffs are given in Table~\ref{Tab:SimpleExample}. The disturbance process is uniformly distributed on $\{0,1\}$. It is easy to verify that $\rho = 0.1$ (from~\eqref{Eq:RhoDef}), and that the interdependence property (from Definition~\ref{Def:Interdependence}) is satisfied, for $G_2$.

We simulated $10^{6}$ iterations of the algorithm~\eqref{Eq:AgentDynamics}--\eqref{Eq:StateDynamics} in Matlab, with a time-varying $\varepsilon$-sequence, and $c=2$. The experimentation rate was modified by setting $\varepsilon_{k+1} = 0.99995 \varepsilon_{k}$, with an initial value of $\varepsilon_{1} = 0.1$. The results of a typical sample run of our simulation are presented in Table~\ref{Tab:Hist3}, validating the results of Theorem~\ref{Thm:ParetoEffLearning}. The average welfare over all the iterations was $1.6937$.

In Table~\ref{Tab:Hist3}, we have displayed a list of states and the normalised number of occurrences of these states, only when this figure was larger than $0.0001$. This is because a total of $101$ states were explored by this simulation. Note that some of the states, such as $\{ [2,0.85,\mathcal{C}] , [1,0.80,\mathcal{C}] \}$ are examples of states in $\mathds{C}^1$. 
\end{example}

The average welfare obtained from playing the optimal action profile(s) will, in general, be different from $\mathcal{W}^*$, the maximum welfare in~\eqref{Eq:WelfareMax}. This is because the optimal action profile maximises the welfare for the most favourable value of the disturbance as per Theorem~\ref{Thm:ParetoEffLearning}. When averaged over all possible values of the disturbance, the welfare will lie in the interval $[\mathcal{W}^* - n \rho , \mathcal{W}^* + n \rho]$, depending on $\Pr_w$. For the above example, the average welfare equals $\mathcal{W}^*$, which may not always be the case as we see in the next example.

\begin{figure*}[th]
	\centering
		\includegraphics[width=\textwidth]{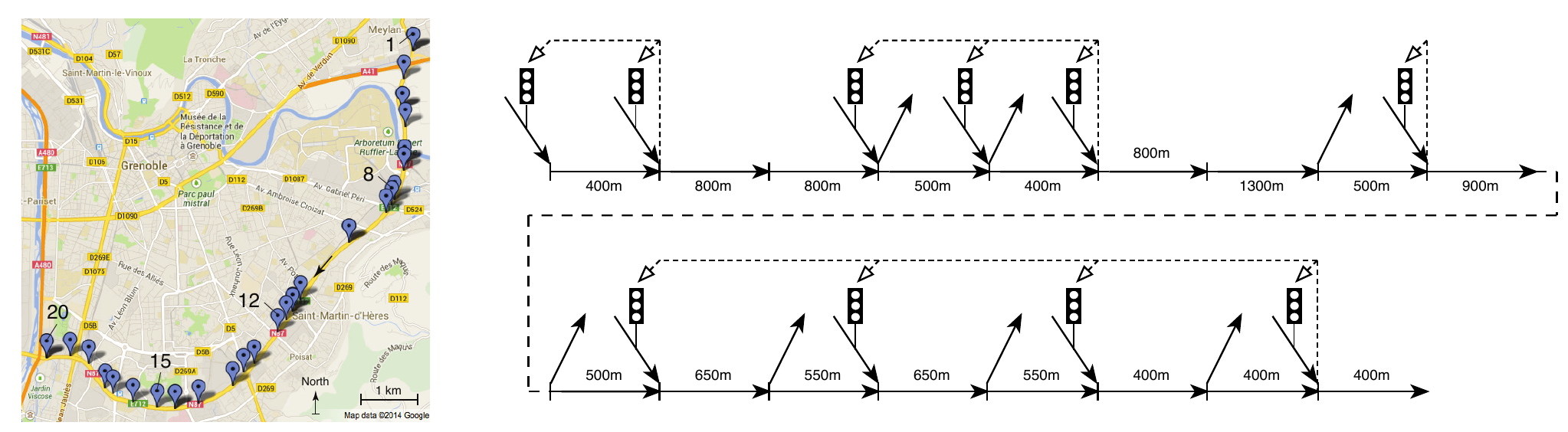}
	\caption{Map of the Grenoble South Link as depicted in~\cite{de2015grenoble} and the corresponding freeway topology. Also shown is a particular coordination pattern between onramps, corresponding to the action profile
	 $\footnotesize [ \text{COR}, \text{LOC}, \text{COR}, \text{COR}, \text{LOC}, \text{LOC}, \text{COR}, \text{COR}, \text{COR}, \text{LOC} ]^{\top}$.}
	\label{fig:grenoble}
\end{figure*}

\begin{example} \label{Ex:RampCoordination}
In freeway traffic control, one seeks to estimate the occupancy of a freeway, usually via loop detectors~\cite{Gibbens2008}, and subsequently adjust speed limits~\cite{Hegyi2005} or traffic lights on the onramps~\cite{papageorgiou2000freeway}, a technique known as ramp metering, to improve traffic flow. However, popular freeway traffic models such as the cell transmission model~\cite{Daganzo1994,daganzo1995cell} or the Metanet model
~\cite{messner1990metanet}, see also~\cite{Payne1971} for an overview of traffic models, are highly nonlinear and methods to design controllers for traffic networks often do not scale well. To reduce the communication and computation burden, distributed controllers that act on mostly local information are required. Model-based approaches for decomposition exist~\cite{Dunbar2006,Oliveira2010}, but in fact, local feedback~\cite{papageorgiou1991alinea} and the combination of local feedback and heuristic, high-level coordination~\cite{papamichail2010heuristic} are among the most popular and practically successful strategies.
To demonstrate its versatility, we will evaluate the efficacy of our method in learning a ramp coordination pattern, for given local controllers.

Consider a freeway with a number of onramps. The idea of ramp metering is to control the traffic inflow from the ramps via traffic lights so as to avoid congestion on the mainline. Both theoretical~\cite{gomes2006optimal} and practical~\cite{papageorgiou2003review} studies have demonstrated that this approach can potentially avoid traffic breakdown in congestion and reduce the sum of travel times of all drivers (TTS, Total Time Spent). An effective metering strategy is to control the inflow such that the local traffic density does not exceed the threshold to congestion, the so-called critical density~\cite{papageorgiou1991alinea}. However, there are limits to this strategy. Multiple ramps, each controlling the local traffic densities, are coupled by the mainline flow as traffic travels downstream and congestion queues can spill back upstream. If no control action of a single ramp is sufficient to prevent congestion of an adjacent bottleneck, then coordination between ramps may hold the answer~\cite{papamichail2010heuristic}.

In this example, we aim to learn a coordination pattern, while the low-level metering policy remains fixed. We consider ten ramps on a freeway located in Grenoble, as presented in~\cite{de2015grenoble} and depicted in Figure~\ref{fig:grenoble}. We allow for ramps either to control only the local traffic density (LOC) or to coordinate with downstream ramps (COR) and control the ramp occupancy, i.e.\ the queue length divided by the ramp length, according to the occupancy of the next downstream ramp. Therefore, the action set for every agent, i.e.,\ every ramp $i$ is $\mathcal{A}_i = \{ \text{LOC}, \text{COR} \}$.
The utility is computed by simulations of the freeway using the modified cell-transmission model as described in~\cite{karafyllis2014global}, which uses a non-monotonic demand function to model the capacity drop empirically observed in a congested freeway. The local utility for agent $i$ is computed as the sum of the total travel time of all cars in the adjacent section of the freeway and the total waiting time in the onramp queue, which is then mapped to the interval $[0,1]$ via a linear transformation. The utilities do not only depend on the action profile but also on the traffic demand, which acts as an external disturbance. We consider real traffic demands during peak hours of the weekdays May 11$^{\text{th}}$ - May 15$^{\text{th}}$, and hence, the disturbance set is $\mathbb{W} = \{ \text{Mon}, \text{Tue}, \text{Wed}, \text{Thu}, \text{Fri} \}$.

We do not try to identify $\rho$ as per~\eqref{Eq:RhoDef} or verify the interdependence property in this example. Instead, we simply choose $\rho$ to be sufficiently large to ensure convergence of the algorithm within a reasonable number of iterations. We ensure that interdependence holds by complimenting the interaction graph with a communication graph, as suggested in~\cite{Menon2013b}. Each agent broadcasts the mood it computes in~\eqref{Eq:StateDynamics}. It then receives all the other agents' moods and performs the following update step to finalize its own mood, as per
\begin{align}
&m_{i,k+1} = \begin{cases}
\mathcal{D} & \tilde{m}_{i,k+1} = \mathcal{D} \\
\mathcal{C} & \tilde{m}_{j,k+1} = \mathcal{C} , \; \forall j \in N \\
\begin{rcases}
\mathcal{C} & \; \text{w.p.} \; \varepsilon^\beta \\
\mathcal{D} & \; \text{w.p.} \; 1-\varepsilon^\beta
\end{rcases}
& \rm{otherwise}
\end{cases} \label{Eq:MoodDynamics}
\end{align}
where $\tilde{m}_{i,k+1}$ is the mood of the $i^{\textrm{th}}$ agent updated locally as per~\eqref{Eq:StateDynamics}. The above update compliments the interaction between the agents by coupling the moods, and is controlled by the parameter $\beta$. If each agent broadcasts its mood to all other agents, this update alone will suffice to ensure the interdependence property, irrespective of the utility functions.
Thus, all the results in this paper, including Theorem~\ref{Thm:ParetoEffLearning}, can be shown to hold for this modified algorithm, for $\rho$ chosen as per~\eqref{Eq:RhoDef}. However, in a real-world setting it might be difficult to compute a suitable bound on $\rho$ beforehand. Instead, we chose $\rho$ empirically to facilitate quick convergence, sacrificing the guarantees that come with Theorem~\ref{Thm:ParetoEffLearning}. The performance is then checked a posteriori.

We simulated $1000$ iterations of the algorithm~\eqref{Eq:AgentDynamics}--\eqref{Eq:StateDynamics},~\eqref{Eq:MoodDynamics}, in Matlab, with $\varepsilon=0.0001$, $c=10$, $\beta=0.00005$ and $\rho=0.6$. The algorithm explored $36$ different action profiles before settling on the ramp coordination schedule 
$\footnotesize \begin{bmatrix} \text{COR}, \text{LOC}, \text{LOC}, \text{COR}, \text{COR}, \text{COR}, \text{LOC}, \text{LOC}, \text{COR}, \text{LOC} \end{bmatrix}^{\top}$. The corresponding baseline utility was $9.6$ and the algorithm spent $890$ out of $1000$ iterations in the above state. The average utility obtained over the entire simulation run was $8.72$, in comparison to an average utility of $6.30$ for the uncontrolled case.
In terms of travel times, this corresponds to savings of $31\%$ over the uncoordinated case. Note that we compute the savings just for the rush-hour period and therefore this value might exceed the savings typically reported for ramp metering field trials, which are usually computed for the entire day~\cite{papageorgiou2000freeway}.
\end{example}

\section{Conclusions} \label{Sec:Conclusions}

We presented a distributed learning algorithm, based on the algorithm in~\cite{Marden2014}, that can be used to learn Pareto-efficient solutions in the presence of disturbances. Our algorithm learns efficient action profiles corresponding to the most favourable disturbance, and specifies a range for the average welfare. In general, the approach outlined in this paper is particularly well suited to problems where the disturbances can be modelled as a finite set of small perturbations from a nominal model. Our examples validated the main result in our paper, and also illustrated the potential of this randomised approach. In many applications, the average welfare is an important performance metric. In future work, we wish to explore randomized approaches that optimize the average welfare obtained.

\bibliographystyle{ieeetr}

\appendix

\section{Proof of Lemmas} \label{App:ProofLemmas}

\begin{proof}[Proof of Lemma~\ref{Lem:RegPert}]
We first show that Property~i holds under $\Pr^{\varepsilon}$ for $\varepsilon >0$. 
Note that all states are accessible from any state $d \in \mathcal{D}$, from the definition of the state space $\mathbb{Z}$. In other words, there exists an integer $\tau>0$ such that $\Pr(z_{k+\tau}=z' | z_{k}=d) > 0$ for all $z' \in \mathbb{Z}$. Next, note that both content and discontent agents chooses an action with a probability distribution that is fully supported on $\mathbb{A}_{i}$, as per~\eqref{Eq:AgentDynamics}. Due to this, one or more agents can become discontent. Along with the interdependence property, this ensures that all agents can consequently become discontent, and thus, a state such as $d$ is accessible from any other state. In other words, there exists an integer $\tau'>0$ such that $\Pr(z_{k+\tau'}=d | z_{k}=z') > 0$ for all $z' \in \mathbb{Z}$. This proves that the Markov chain is irreducible. Furthermore, many of these states permit a return to the same state with some positive probability, i.e., there exist states $z \in \mathbb{Z}$ such that $\Pr(z_{k+1}=z | z_{k}=z) >0$. These states are aperiodic, which in combination with the irreducibility property effectively renders the Markov chain aperiodic.

By inspection of~\eqref{Eq:AgentDynamics}--\eqref{Eq:StateDynamics}, it is clear that Property~ii is satisfied. We now show that Property~iii holds. Note that the transition probabilities contain terms with exponents of $\varepsilon$ or its complement. The resistance $r=0$ for transition probabilities containing complements of $\varepsilon$, because these transitions occur under $\Pr^{0}$. All other transition probabilities contain negative exponents of $\varepsilon$ resulting in positive resistances, as required by Property~iii.
\end{proof}

\begin{proof}[Proof of Lemma~\ref{Lem:RecurrenceClasses}]
Consider a state $z \in \mathds{D}$. Under $\Pr^0$, each agent picks an action with uniform probability and no utility ever makes an agent content. Thus, any accessible state remains in $\mathds{D}$.

Consider a singleton state in any of the classes $\mathds{C}^{0}$ or $\mathds{C}^{m}$, for $0 < m < n$. Under $\Pr^0$, each agent plays the same action again, and the utilities received by the agents satisfy the interval rule. Thus, the agents remain in the same state.

In the $\mathds{C}^{m}$ states, for $0 \le m < n-1$, when a subset of agents $J \subset N$ choose a new action and become content, there are two circumstances under which the new state is not recurrent. If a value of the disturbance results in a payoff that does not satisfy the interval rule, the corresponding agent(s) become discontent and a new state is reached, which has a mix of content and discontent agents. Under $\Pr^0$, a discontent agent remains discontent. Furthermore, due to the interdependence property, a subset of discontent agents will cause at least one content agent to violate the interval rule and become discontent. This repeats until all agents are discontent, thus reaching $\mathds{D}$. A similar situation occurs if the new action causes the payoff received by any agent in $N \setminus J$ to fall outside the prescribed interval. Thus, in general, no state with a mix of content and discontent agents, is recurrent.
\end{proof}

\begin{proof}[Proof of Lemma~\ref{Lem:StochPotC0}]
First we show that $\gamma(z^{0})$ is less than or equal to the right hand side in~\eqref{Eq:StochPotC0}, and then we show the reverse, thus proving the equality relationship in the lemma.

To show the first path, construct the following tree $T$ rooted at $z^{0}$: Add a directed link $z^{0'} \rightarrow d$ with resistance $c$ between every $z^{0'} \in \mathds{C}^{0} \setminus z^{0}$ and some $d \in \mathds{D}$. Then, add a directed link $z^{m} \rightarrow d$ with resistance $c$ between every $z^{m} \in \mathds{C}^{m}$, for $0 < m < n$, and some $d \in \mathds{D}$. Finally add a directed link $d \rightarrow z^{0}$ from some $d \in \mathds{D}$, with resistance $\sum_{i \in N} (1 - \bar{u}_{i}(z^{0}) )$. The resistance of this tree $\gamma(T) = c (\sum_{m=0}^{n-1} |\mathds{C}^{m}| -1) + \sum_{i \in N} (1 - \bar{u}_{i}(z^{0}) )$, thus establishing that $\gamma(z^{0}) \le \gamma(T)$.

To show the reverse, consider a general tree $T'$ rooted at $z^{0}$. It may differ from $T$ in one or more of the following aspects:
\begin{enumerate}[label={(\alph*)}]
\item It may contain a path of length $q$ between the discontent class and the singleton $z^{0}$, such as $d \rightarrow z^{m_{1}}_{1} \rightarrow \dots z^{m_{q}}_{q} \rightarrow z^{0}$, where $d \in \mathds{D}$, $z^{m_{1}}_{1} \in \mathds{C}^{m_{1}}$, $\dots$, $z^{m_{q}}_{q} \in \mathds{C}^{m_{q}}$.
\item It may contain paths of length $s$ between a singleton from any of the $m^{\textrm{th}}$-content classes and the discontent class, such as $z^{m} \rightarrow z^{m_{1}}_{1} \rightarrow \dots z^{m_{s}}_{s} \rightarrow d$, where $z^{m} \in \mathds{C}^{m}$, $z^{m_{1}}_{1} \in \mathds{C}^{m_{1}}$, $\dots$, $z^{m_{s}}_{s} \in \mathds{C}^{m_{s}}$ and $d \in \mathds{D}$.
\end{enumerate}
From Table~\ref{Tab:Resistances}, we note that the path of length $q$ in case~(a) has a resistance $r(d \rightarrow z^{m_{1}}_{1} \rightarrow \dots z^{m_{q}}_{q} \rightarrow z^{0}) \ge qc + \sum_{i \in N} (1 - \bar{u}_{i}(z^{0}))$. Construct a tree $T_{(a)}$ by replacing this path in $T'$ with a set of links $z^{m_{i}}_{i} \rightarrow d$, for $1 \le i \le q$, and $d \rightarrow z^{0}$. By making these changes, we obtain a total resistance of $qc + \sum_{i \in N} (1 - \bar{u}_{i}(z^{0}))$. 
Thus, we have constructed a tree with $\gamma(T_{(a)}) \le \gamma(T')$.

Next, note that the path in case~(b), $z^{m} \rightarrow z^{m_{1}}_{1} \rightarrow \dots z^{m_{s}}_{s} \rightarrow d$ has a resistance $r \ge (s+1)c$. Construct a tree $T_{(b)}$ by replacing the links in the path with the links $z^{m} \rightarrow d$ and $z^{m_{i}}_{i} \rightarrow d$, for $1 \le i < s$, each of resistance $c$. 
Then, $\gamma(T_{(b)}) \le \gamma(T')$. Thus, we have shown that $\gamma(T^{*}) = \gamma(T) \le \gamma(z^{0})$.
\end{proof}

\begin{proof}[Proof of Lemma~\ref{Lem:NotStochStable}]
Consider a tree rooted at $\mathds{D}$. Then, it must contain a link $z^{0} \rightarrow d$, for some $d \in \mathds{D}$, of resistance $c$. Replace it with the link $d \rightarrow z^{0}$ that incurs a lesser resistance $\sum_{i \in N} (1 - \bar{u}_{i}(z^{0})) \le n < c$. Thus, we have constructed a tree rooted at $z^{0}$ with lower potential.

Similarly, consider a tree rooted at $z^{m} \in \mathds{C}^{m}$, for $0 < m < n$. This tree must contain a path $d \rightarrow z^{0} \rightarrow z^{1} \rightarrow \dots \rightarrow z^{m}$, with resistance $r \ge mc + \sum_{i \in N} (1 - \bar{u}_{i}(z^{0}))$. Replace the links in the path $z^{0} \rightarrow z^{m}$ with the links $z^{l} \rightarrow d'$, for $0 < l \le m$ and $d' \in \mathds{D}$, which results in a resistance of exactly $mc + \sum_{i \in N} (1 - \bar{u}_{i}(z^{0}))$. Thus, we have constructed a tree rooted at $z^{0}$ with lower stochastic potential, and proved our result.
\end{proof}

\section{Calculation of Resistances} \label{App:ResistanceCalc}

In this section, we use $d$, $z^{0}$ and $z^{m}$ to denote a state $d \in \mathds{D}$ and singleton states $z^0 \in \mathds{C}^0$ and $z^m \in \mathds{C}^m$, respectively.

Let us begin with row~$1$ of Table~\ref{Tab:Resistances}. The transition $d \rightarrow z^0$ occurs only when all agents are content with the received payoffs, which happens with probability $\prod_{i \in N} \varepsilon^{1-u_i}$. This gives us the resistance $r_{d z^0}$ in row~$1$. The resistance $r_{d z^m}$ in row~$2$ follows from the definition of a resistance between recurrence classes in~\eqref{Eq:ResistanceRecClass} and the fact that the transition $d \rightarrow z^{m}$ only occurs through a state $z^0$.

For the transition $z^0 \rightarrow d$ to occur, at least one content agent must experiment and become discontent, which happens with probability of order $O(\varepsilon^c)$. Then, all agents become discontent eventually. Thus, $r_{z^0 d} = c$ in row~$3$. The same holds for the transition $z^m \rightarrow d$ in row~$4$.

In row~$5$, the transition $z^0_1 \rightarrow z^0_2$ between the singleton states $z^0_1, z^0_2 \in \mathds{C}^0$ can occur in multiple ways. The transition with the least resistance occurs when an agent experiments with a new action and becomes content with the payoff it receives, while not affecting the payoffs of any other agent. Thus, $r_{z^0_1 z^0_2} \ge c + 1-\bar{u}_i(z^0_2) \ge c$. In general, however, this transition occurs through intermediate states in $\mathds{D}$, resulting in $r_{z^0_1 z^0_2} \le c + \sum_{i \in N} 1 - \bar{u}_i(z^0_2) \le c+n \le 2c$. Transitions through states in $\mathds{C}^m$ are not considered because these incur resistances of greater than $2c$.

A similar argument can be applied to calculate the resistance $r_{z^m_1 z^m_2}$ of a transition between two singleton states $z^m_1, z^m_2 \in \mathds{C}^m$ in row~$6$. When the transition only requires one agent to experiment and be content with a new action, $r_{z^m_1 z^m_2} \ge c + 1-\bar{u}_i(z^m_2) \ge c$. Other transitions occur through intermediate states in $\mathds{D}$, resulting in
\begin{align*}
r_{z^m_1 z^m_2} \le \min_{z^0 \in \mathds{C}^0, z^{1} \in \mathds{C}^{1}, \dots,z^{m-1} \in \mathds{C}^{m-1}} &c + \sum_{i \in N} 1 - \bar{u}_i(z^0) 
+ r_{z^0 z^{1}} + \dots + r_{z^{m-1} z^m_2}
\end{align*}
The worst case least resistance path $z^{0} \rightarrow z^{m}_{2}$ occurs when $n-1$ agents experiment and become content to ensure $z^{0} \rightarrow z^{1}$, $n-2$ agents experiment and become content to ensure $z^{1} \rightarrow z^{2}$ and so on until $n-m$ agents experiment and become content to ensure $z^{m-1} \rightarrow z^{m}_{2}$. This can happen when intermediate states, which are required to ensure that $z^{0} \rightarrow z^{m}_{2}$ occurs with fewer experimenting agents, are not recurrent. Thus, we get
\begin{align*}
r_{z^m_1 z^m_2} \le &\min_{z^0 \in \mathds{C}^0, z^{1} \in \mathds{C}^{1}, \dots,z^{m-1} \in \mathds{C}^{m-1}} & &c + n - \sum_{i \in N} \bar{u}_{i}(z^{0}) \\
& & &+ \underbrace{(n-1)c + \sum_{j_{1} =1}^{n-1} 1-\bar{u}_{j_{1}}(z^{1})}_{r_{z^{0} z^{1}}} + \dots 
+ \underbrace{(n-m)c + \sum_{j_{m} =1}^{n-m} 1-\bar{u}_{j_{m}}(z^{m}_{2})}_{r_{z^{m-1} z^{m}_{2}}} \\
\le & & &c + n + (n-1)c + n-1 + \dots + (n-m)c + n-m
\end{align*}
Using the fact that $c > n$ and summing over the series, we obtain the upperbound in~\eqref{Eq:rmm}. Again, transitions through states in $\mathds{C}^l$, for $l \neq m$, are not considered because these may incur resistances of greater than $\bar{r}_{m}$.

Similar arguments can be used to calculate the resistance in row~$7$ of a transition between two singleton states $z^l \in \mathds{C}^{l}$ and $z^{m} \in \mathds{C}^{m}$, for $0 \le l,m < n$ and $m \neq l$. The least resistant paths require $|m-l|$ agents to experiment and become content, and other transitions occur through an intermediate state in $\mathds{D}$. Transitions through other states in $\mathds{C}^{s}$, for $1 \le s < n$, are not considered as the resistances of such paths can be higher than $\bar{r}_{m}$.

\end{document}